\documentclass[%
reprint,
mathtools,
amsmath,
amssymb,
aps,
pra,
superscriptaddress,
usenames,dvipsnames
]{revtex4-2}

\usepackage{graphicx}% Include figure files
\usepackage{dcolumn}% Align table columns on decimal point
\usepackage{bm}% bold math
\usepackage[colorlinks=true,linktocpage=true]{hyperref}% add hypertext capabilities
\hypersetup{
    allcolors  = {blue},
}
\usepackage{soul}

\usepackage{amsthm}
\usepackage{nicefrac}
\newtheorem{theorem}{Theorem}[section]
\newtheorem{lemma}[theorem]{Lemma}
\usepackage{tikz}
\usetikzlibrary{quantikz}

\usepackage{xcolor}

\begin{document}

\title{Pauli-based model of quantum computation with higher-dimensional systems}

\author{Filipa C. R. Peres}
\email{filipa.peres@inl.int}
\affiliation{International Iberian Nanotechnology Laboratory (INL), Av. Mestre Jos\'{e} Veiga, 4715-330 Braga, Portugal.}
\affiliation{Departamento de F\'{i}sica e Astronomia, Faculdade de Ci\^{e}ncias, Universidade do Porto, rua do Campo Alegre s/n, 4169–007 Porto, Portugal.}

\date{\today}

\begin{abstract}
Pauli-based computation (PBC) is a universal model for quantum computation with qubits where the input state is a magic (resource) state and the computation is driven by a sequence of adaptively chosen and compatible multiqubit Pauli measurements. Here we generalize PBC for odd-prime-dimensional systems and demonstrate its universality. Additionally, we discuss how any qudit-based PBC can be implemented on actual, circuit-based quantum hardware. Our results show that we can translate a PBC on $n$ $p$-dimensional qudits to adaptive circuits on $n+1$ qudits with $O\left( pn^2/2 \right)$ \textsc{sum} gates and depth. Alternatively, we can carry out the same computation with $O\left( pn/2\right)$ depth at the expense of an increased circuit width. Finally, we show that the sampling complexity associated  with simulating a number $k$ of virtual qudits is related to the robustness of magic of the input states. Computation of this magic monotone for qutrit and ququint states leads to sampling complexity upper bounds of, respectively, $O\left( 3^{ 1.0848 k} \epsilon^{-2}\right)$ and $O\left( 5^{ 1.4022 k} \epsilon^{-2}\right)$, for a desired precision $\epsilon$. We further establish lower bounds to this sampling complexity for qubits, qutrits, and ququints: $\Omega \left( 2^{0.5431 k} \epsilon^{-2} \right)$, $\Omega \left( 3^{0.7236 k} \epsilon^{-2} \right)$, and $\Omega \left( 5^{0.8544 k} \epsilon^{-2} \right)$, respectively.

\begin{description}
\item[Key words]
Pauli-based computation; magic-state injection; hybrid computation; qudits.
\end{description}
\end{abstract}

\maketitle

\section{Introduction}\label{sec:Introduction}

The quantum circuit model is probably the most pervasive and well-known paradigm for universal quantum computation~\cite{NielsenChuang}. The fact that it is a fairly intuitive model inspired by a classical analog has likely contributed to its substantial success, as it provides a transparent framework for theoretical studies of quantum computation and quantum complexity compared to the classical case. Nevertheless, other models exist that have no classical counterpart, for instance, the one-way model (1WM) of measurement-based quantum computation (MBQC)~\cite{RaussBrie2001}, fusion-based quantum computation~\cite{FBQC2023}, and Pauli-based  computation (PBC)~\cite{BSS2016}.

Research into these models (and quantum computing and technologies in general) has been largely focused on two-level systems, called \textit{qubits}. However, thinking about higher-dimensional systems, dubbed \textit{qudits}, is not unreasonable, since we usually enforce a two-level structure on quantum systems that are naturally multileveled (e.g., atomic energy levels). Additionally, fewer qudits are needed to encode a Hilbert space of the same size, certain quantum algorithms showcase improved implementations when using qudits~\cite{RRGil2007, Gedik+2015}, and magic-state distillation seems more efficient in higher-dimensional systems~\cite{CampbellAB2012}. Research on qudits has also seen interesting results in the context of quantum foundations~\cite{Kasz+2002, CollinsGisin+2002} and quantum cryptography~\cite{BechPasq2000, Machiavello2002}. Moreover, successful experimental realizations with qudits have emerged in recent years~\cite{Lanyon+2008, Lanyon+2009, Bianchetti+2010, Senko+2015, LuoZhong+2019, HuZhang+2020, BlokRam+2021, ChiHuang+2022, Ringbauer2022} as well as work on benchmarking such quantum hardware~\cite{Morvan+2021}. Thus, it is conceivable that the future of quantum technology relies on qudits rather than qubits. 

Focusing on quantum computing, the quantum circuit model has been proven to be universal for fault-tolerant computation with qudits~\cite{Gottesman1999}, and afterward, the same was done for the 1WM~\cite{ZZXS2003, Clark2006}. In this paper, we focus on generalizing the Pauli-based model of quantum computation beyond the qubit case.

In addition to its intrinsic theoretical relevance and the aforementioned motivations underpinning the use of higher-dimensional systems, we hope that this work calls more attention to some of the peculiarities of PBC and arouses the interest of the quantum computing community. Indeed, unlike other universal models, PBC remains a rather understudied framework. Here, we extend its realm of applicability, adding to very recent (and still scarce) research employing PBC in a myriad of contexts ranging from fault-tolerant implementations~\cite{Litinski2018latticesurgery, Chamberland+2022catcodes} and classical simulation of quantum computation~\cite{Zurel2020, Zurel2023}, to circuit compilation and hybrid computation~\cite{PeresGa2022}.

A detailed discussion of PBC is deferred to Sec.~\ref{subsec: PBC with qubits}. For now, it suffices to know that, as originally formulated in Ref.~\cite{BSS2016}, this model takes as input a separable magic state of $n$ qubits, $\left|T \right>^{\otimes n}$, with $\left|T \right> = (\left|0 \right> + e^{i\pi/4}\left|1 \right>)/\sqrt{2},$ and is driven by a sequence of (at most) $n$ independent and pairwise commuting $n$-qubit Pauli measurements. One of the main drawbacks of PBC is probably the nonlocal nature of these measurements. For this reason, besides generalizing the model for odd-prime-dimensional qudits and demonstrating its universality, we also discuss how a sequence of multiqudit Pauli observables can be measured in practice in circuit-based quantum hardware. Specifically, we propose two different methods to do so; the first requires a single auxiliary qudit and leads to a quantum circuit with depth quadratic in $n$, whilst the second uses $n$ auxiliary qudits to achieve a linear depth.

Due to the larger Hilbert space achieved with higher-dimensional systems, we expect PBC to require fewer qudits as $p$ increases. As a consequence, the total number of measurements and also, potentially, the number of qudits involved in each measurement are reduced; this should lead to simpler and shallower implementations and help to relax the demands on the coherence time of the qudits.

One of the most promising features of PBC is the possibility to remove a certain (desired) number of qubits (dubbed \textit{virtual qubits} in Ref.~\cite{BSS2016}) from the computation. It was shown that the cost of offloading $k$ virtual qubits to a classical machine is a sampling complexity that scales exponentially with $k$~\cite{BSS2016, PeresGa2022}. We demonstrate that this hybrid quantum-classical scenario can be generalized for odd-prime-dimensional qudits and derive upper and lower bounds for the corresponding sampling complexities for qubits, qutrits ($p=3$), and ququints ($p=5$). Remarkably, the results suggest that the cost of extending the quantum memory by $k$ virtual qudits increases by an amount larger than what can be explained by the increase of the qudit's dimensionality alone.

This paper is structured as follows. In Sec.~\ref{sec: Background}, we start by presenting all the background required for the understanding of our work. Sec.~\ref{sec: PBC for qudits} contains our results; in Sec.~\ref{subsec: Proof of universality}, we show that PBC is universal for quantum computation with qudits, in Sec.~\ref{subsec: Practical implementation}, we describe how measurements of $n$-qudit Pauli operators can be implemented in practice, and, in Sec.~\ref{subsec: Hybrid computation}, we assess the cost of performing hybrid quantum-classical computation within this framework. We conclude the paper with some final remarks in Sec.~\ref{sec:Conclusions}.

\section{Background}\label{sec: Background}

To frame our results in the context of previous findings, we start this section with a state-of-the-art review of PBC with qubits. Then, we present generalized versions of the Pauli and Clifford groups for qudits and explain how to perform universal quantum computation with such higher-dimensional systems by implementing non-Clifford gates via magic-state injection. We focus on the case of $p$-level qudits, where $p$ is an odd-prime number.

\subsection{Pauli-based computation with qubits}\label{subsec: PBC with qubits}

PBC is a model for universal quantum computation with qubits that has seen little research since its proposal by Bravyi, Smith, and Smolin back in 2016~\cite{BSS2016}. According to its original formulation, PBC is driven by a sequence of (at most) $n$ (adaptively chosen) measurements of independent and compatible Pauli observables on $n$ qubits initialized in the input state $\left|T \right>^{\otimes n} = \left( \left| 0 \right> + e^{i \pi /4}\left| 1 \right> \right)^{\otimes n} / 2^{n/2}.$ Since its computational steps are measurements, it is clear that, much like the 1WM, PBC is comprised within the wide field that is MBQC. Thus, the need for adaptive measurements and classical feed-forward comes as no surprise, since this is an intrinsic characteristic of measurement-based paradigms: Owing to the non-deterministic nature of measurements in quantum mechanics, adaptivity is needed to ensure that the overall quantum computation is still deterministic. 

Besides their shared similarities, it is also interesting to examine the differences between the 1WM and PBC. Notably, the first takes as input an entangled, stabilizer resource state; during the computation, the qubits of this state are subjected to single-qubit (magic or stabilizer) measurements which successively expend the entanglement. On the other hand, PBC uses a separable, magic-state input that is subjected to a sequence of entangling multiqubit Pauli measurements which sequentially spend the magic resource. We understand that whilst the 1WM brings to the foreground the importance of entanglement for universal quantum computation, PBC highlights the role of magic.

A seemingly promising idea is that of casting PBC in a more general light, stepping away from this clear-cut separation of the magic and entanglement resources and considering the use of other input states. For instance, Chamberland and Campbell \cite{ChamberlandCampbell2022} noted that, instead of considering the $\left< \text{Clifford},T\right>$ gate set, it is advantageous to consider an overcomplete set such as, for instance, $\left< \text{Clifford},T,\text{Toffoli}\right>$. In the PBC paradigm, this corresponds to allowing both $\left| T \right>$ and $\left| \mathrm{Toff} \right> = (\left|000 \right> + \left|010 \right> + \left| 100\right> + \left| 111\right>)/2$ magic-state inputs~\cite{Zhou2000}. In doing so, it should be possible to find the solution to a given task via a PBC which requires fewer qubits and Pauli measurements than would be possible if only $\left| T \right>$ magic states were allowed. Exploring these variations to the original PBC framework could unlock resource savings relevant to the current noisy intermediate-scale quantum regime.

It was noted in the Introduction that perhaps one of the main drawbacks of PBC is the non-destructive and multiqubit nature of its measurements. A simple way of understanding how to perform these Pauli measurements is to translate things back to the circuit model. This was the approach followed in Ref.~\cite{PeresGa2022} where the authors were mainly concerned with near- to intermediate-term applications, discussing things at the physical level. This is also the approach taken in Sec.~\ref{subsec: Practical implementation}. Meanwhile, thinking in long-term fault-tolerant quantum computing, we note that stabilizer codes provide a natural way of performing non-destructive measurements of nonlocal Pauli  observables~\cite{Litinski2018latticesurgery, Chamberland+2022catcodes, ChamberlandCampbell2022}. In this paradigm, the computation can be divided into two separate tasks. First, that of the offline preparation of the resource states in magic-state factories, taking a time $t_M$; and second, the actual PBC sequence, happening in a time-frame $t_{\text{PBC}}$~\cite{ChamberlandCampbell2022}. In principle, $t_M$ can be made as small as desirable by increasing the number of magic-state factories, which means that $t_{\text{PBC}}$ should determine the overall runtime. This motivates finding ways of improving $t_{\text{PBC}},$ for instance, by (somehow) parallelizing the computation. Unfortunately, such strategies often come at the expense of a larger number of qubits~\cite{Litinski2019gameofsurfacecodes, ChamberlandCampbell2022}. Increasingly more work is arising that focuses on understanding these space-time trade-offs in fault-tolerant architectures, and PBC seems to be a natural way of thinking about these techniques.

PBC has also been explored in the context of circuit compilation and hybrid quantum-classical computation. As illustrated in Ref.~\cite{PeresGa2022}, for some families of quantum circuits dominated by Clifford gates, PBC provides a way of reducing the quantum resources needed for the computation. It is also a framework wherein hybrid computation can be implemented straightforwardly allowing for the simulation of a certain number, $k$, of virtual qubits at the cost of a sampling complexity that scales as $N=O(2^{0.7374k})$. Sec.~\ref{subsec: Hybrid computation} is focused on exploring this hybrid computation setting further. There, we derive upper bounds for the sampling complexity for qutrits and ququints; we also establish lower bounds for this quantity, something that, to the best of our knowledge, was missing from the literature up to this point.

Finally, it is worth noting that PBC has also made some appearances in recent works related to quantum complexity and classical simulation of quantum circuits~\cite{YogaJS2019, Zurel2020, Zurel2023, Kocia2022}. The usefulness in this context seems to stem from the fact that all the magic is pushed to the preparation of the input state, whilst the computation itself is simply driven by Pauli measurements. For instance, in Ref.~\cite{Zurel2020}, the authors constructed a hidden variable model where (i) any quantum state can be positively represented in the considered phase space and (ii) the update rule of quantum states under Pauli measurements is probabilistic. Because Pauli measurements are sufficient for universal quantum computation, this proves that their hidden variable model describes all possible quantum computations.

This review highlights how interest in PBC is starting to arise in some contexts like fault tolerance, classical simulation, circuit compilation, and hybrid computation. These works, while sharing some common ground in their use of this computational model, have mostly occurred disjointedly from one another. Casting our eyes to the future, it is important to start connecting these (currently separate) contributions, understanding how they can be used to enhance one another, and creating a clearer picture of PBC, its peculiarities, benefits, and drawbacks compared to other models for universal quantum computation. Additionally, expanding the reach of PBC into other fields and contexts should shape new perspectives and originate interesting new results.

\subsection{Generalized Pauli and Clifford groups}\label{subsec: Pauli + Clifford groups}

The single-qubit Pauli operators $X$ and $Z$ can be easily generalized for $p$-dimensional qudits in the following manner~\cite{Gottesman1999}:
\begin{equation}
    X \left| j \right> = \left| j + 1\right> \text{ and } Z \left| j \right> = \omega ^{j} \left| j \right> \,,
    \label{eq: generalized X and Z}
\end{equation}
where $\omega = e^{2\pi i / p}$ is the $p$th root of unity and addition is carried out modulo $p$. The Pauli group on one qudit is generated by these operators together with the phase $\omega$: $\mathcal{P}_1= \left< \omega, X, Z\right>$.

Any Pauli operator $P$ on $n$ qudits can be written as:
\begin{equation}
    M = \omega ^{\lambda} X(\boldsymbol{x}) Z(\boldsymbol{z})\,,
    \label{eq: generic Pauli operator/mmt}
\end{equation}
with $\lambda \in \mathbb{F}_p$ and $\boldsymbol{x},\boldsymbol{z} \in \mathbb{F}_p^n$ so that $X(\boldsymbol{x}) = X^{x_1}\otimes \dots \otimes X^{x_n}$ [and equivalently for $Z(\boldsymbol{z})$]. An $n$-qudit state that is a simultaneous +1 eigenstate of $n$ independent and pairwise commuting Pauli operators is known as a stabilizer state.
\begin{table}[]
  \vspace{-0.21cm}\caption{Transformation of the two generators of the single-qudit Pauli group ($X$ and $Z$) under conjugation by the single-qudit Clifford gates $F$ and $S$ (top), and of the four generators of the two-qudit Pauli group ($X\otimes I$, $I\otimes X$, $Z\otimes I$, and $I\otimes Z$) under the action of the two-qudit Clifford unitary $\mathrm{SUM}_{1,2}$ (bottom).}
  \begin{tabular}{>{\centering\arraybackslash}p{0.6cm} >{\centering\arraybackslash}p{0.6cm} >{\centering\arraybackslash}p{0.6cm}}
  \hline\hline
      & $X$  & $Z$           \\ \hline
  $F$ & $Z$  & $X^{\dagger}$ \\
  $S$ & $XZ$ & $Z$   \\
  \hline\hline
  \end{tabular}
  
  \vspace{0.3cm}
  \begin{tabular}{>{\centering\arraybackslash}p{1.2cm} >{\centering\arraybackslash}p{1.2cm} >{\centering\arraybackslash}p{1.2cm} >{\centering\arraybackslash}p{1.2cm} >{\centering\arraybackslash}p{1.2cm}}
                \hline\hline
                 & $X\otimes I$ & $I\otimes X$ & $Z\otimes I$ & $I\otimes Z$          \\ \hline
  $\mathrm{SUM}_{1,2}$ & $X\otimes X$ & $I\otimes X$ & $Z\otimes I$  & $Z^{\dagger}\otimes Z$ \\
  \hline\hline
  \end{tabular}
  \label{tab: Conjugation of the Pauli generations}
\end{table}

The Clifford group is defined as the set of unitary operators $U$ which map the Pauli group, $\mathcal{P}_n$, onto itself under conjugation:
\begin{equation}
    \mathcal{C}_2 = \{ U: \, U \mathcal{P}_n U^{\dagger} = \mathcal{P}_n \}.
\end{equation}

In the qubit setting, it is known that this group is generated by the Hadamard, phase, and controlled-\textsc{not} gates. For $p$-dimensional qudits, the generalization of the Hadamard gate is the so-called Fourier gate~\cite{Gottesman1999}:
\begin{equation}
    F \left| j \right> = \frac{1}{\sqrt{p}} \sum_{k=0}^{p-1} \omega^{j k} \left| k \right>\,.
    \label{eq: Fourier gate}
\end{equation}
Next, the generalization of the phase gate~\cite{Gottesman1999} can be written as
\begin{equation}
    S \left| j \right> = \omega^{j (j-1) 2^{-1}} \left| j \right>\,,
    \label{eq: generalized S gate}
\end{equation}
where $2^{-1}$ is understood to be the inverse of $2$ in the finite field $\mathbb{F}_p$. Finally, the two-qudit \textsc{sum} gate can be regarded as the generalization of the \textsc{cnot} gate~\cite{Gottesman1999}:
\begin{equation}
    \mathrm{SUM}_{1,2} \left| j \right>\left| k \right> = \left| j \right>\left| k + j \right>\,,
    \label{eq: SUM gate}
\end{equation}
with the addition performed modulo $p$.

These gates have been shown to generate any Clifford unitary up to a global phase~\cite{Clark2006}, so that $\mathcal{C}_2 = \left< F, S, \mathrm{SUM}\right>.$ Quantum circuits comprised of stabilizer states, Clifford gates, and Pauli measurements are called stabilizer circuits.

In Table~\ref{tab: Conjugation of the Pauli generations}, we see how the generators of the Pauli group are transformed under conjugation by the Clifford generators. These transformation rules give an efficient way of tracking the evolution of the state of stabilizer circuits. Because of that, these quantum circuits are efficiently simulable in a classical computer; this is the essence of the famous Gottesman-Knill theorem~\cite{GottesmanPhD}.

\subsection{Universality via magic-state injection}\label{subsec: Universality via MSI}

Stabilizer circuits are not universal for quantum computation. To achieve universality, some kind of \textit{magic} (i.e., non-stabilizer) operation is needed. More accurately, in Ref.~\cite{CampbellAB2012}, the authors combined Theorem 7.3 of Ref.~\cite{Nebe2001} and Corollary 6.8.2 of Ref.~\cite{Nebe2006} to note that any magic gate supplementing the Clifford group is enough to achieve universal quantum computation with odd-prime qudits.

Here, we consider the diagonal unitaries $U_v$ defined in Ref.~\cite{HowardVala2012}. For qutrits ($p=3$), we write:
\begin{equation}
    U_v \equiv U_{(v_0,v_1,v_2)} = \sum_{k=0}^{2} \zeta^{v_k} \left| k\right> \left< k \right|\,,
\end{equation}
with $v=\left( 0,\,6z^{\prime}+2\gamma^{\prime}+3\epsilon^{\prime},\, 6z^{\prime}+\gamma^{\prime}+6\epsilon^{\prime} \right)\text{ mod }9$, $z^{\prime},\gamma^{\prime},\epsilon^{\prime}\in \mathbb{F}_3,$ and $\zeta = e^{2\pi i / 9}$ \footnote{Note that, both here and for $p>3$, $\gamma^{\prime}$ must necessarily be different from 0, otherwise the operation is a (diagonal) Clifford. For details see~\cite{HowardVala2012}.}.
For instance, for $z^{\prime}=1,$ $\gamma^{\prime}=2$, and $\epsilon^{\prime}=0$, we have $U_{(0,1,8)} = \left| 0 \right> \left< 0 \right| + e^{2\pi i / 9}\left| 1 \right> \left< 1 \right| + e^{- 2\pi i / 9}\left| 2 \right> \left< 2 \right|.$

For $p>3$, we have:
\begin{equation}
    U_v \equiv U_{(v_0,\ldots,v_{p-1})} = \sum_{k=0}^{p-1} \omega^{v_k} \left| k\right> \left< k \right|\,,
\end{equation}
where $v_k=12^{-1} k \{ \gamma^{\prime} + k \left[ 6z^{\prime} + \gamma^{\prime} \left( 2k - 3\right)\right] \} + k\epsilon^{\prime}$ and $z^{\prime},\gamma^{\prime},\epsilon^{\prime}\in \mathbb{F}_p.$
For instance, for $p=5$ and taking $z^{\prime}=1,$ $\gamma^{\prime}=4$, and $\epsilon^{\prime}=0$, we have $U_{(0,3,4,2,1)} = \left| 0 \right> \left< 0 \right| + e^{-4\pi i / 5}\left| 1 \right> \left< 1 \right| + e^{- 2\pi i / 5}\left| 2 \right> \left< 2 \right| + e^{+4\pi i / 5}\left| 3 \right> \left< 3 \right| + e^{+ 2\pi i / 5}\left| 4 \right> \left< 4 \right|.$

These gates can be implemented via the injection of magic states which we will denote $\left| T_v\right>$ (by analogy with the $\left| T \right>$ magic state for qubits), together with the gadget described in Ref.~\cite{Nadish2020} and depicted in Fig.~\ref{fig: Gadget for the injection of magic gates U_v.}. Protocols for the distillation of the states $\left| T_v\right>=U_{v} \left| + \right>$ were described in Ref.~\cite{CampbellAB2012} and shown to tolerate higher error rates than any protocol described so far for qubits.
\begin{figure}[t]
    \centering
    \begin{tikzpicture}
      \node[scale=1.0] {
        \tikzset{
            my label/.append style={above right,xshift=0.45cm,yshift=-0.25cm}
        }

        \begin{quantikz}[thin lines]
           \lstick{\ket{T_{v}}}  &\qw         &\ctrl{1}             &\gate{U_v \left(X^{\dagger}\right)^{\sigma} U_v^{\dagger}}  &\qw\rstick{$U_v$ \ket{\Psi_{in}}}\\
           \lstick{\ket{\Psi_{in}}}     &\gate{F^2}  &\gate{\mathrm{SUM}}  &\meter{$\sigma$}\vcw{-1}
        \end{quantikz}      
      };
    \end{tikzpicture}
    \caption{Implementation of the magic gate $U_v$ via the gadget described in~\cite{Nadish2020}.}
    \label{fig: Gadget for the injection of magic gates U_v.}
\end{figure}
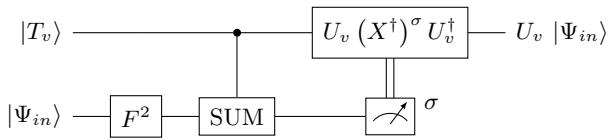

\section{Pauli-based computation with odd-prime qudits}\label{sec: PBC for qudits}

In this section, we generalize the Pauli-based model of quantum computation to $p$-dimensional systems, with $p$ being an odd-prime integer.

\subsection{Proof of universality}\label{subsec: Proof of universality}

We follow a proof analogous to that performed in Ref.~\cite{BSS2016}. We start from a unitary quantum circuit on $n$ qudits, written in terms of the Clifford+$U_v$ (universal) gate set, with $t$ non-Clifford gates $U_v$ and $m$ final $Z$ measurements. Using the gadget in Fig.~\ref{fig: Gadget for the injection of magic gates U_v.}, this circuit can be translated into an adaptive Clifford circuit on $n+t$ qudits initialized in the state $\left| 0 \right>^{\otimes n}\left|T_v\right>^{\otimes t}$, with $t$ intermediate measurements and $m$ final measurements \footnote{The Clifford corrections associated with the gadget in Fig.~\ref{fig: Gadget for the injection of magic gates U_v.} have been calculated in Appendix~\ref{app: Clifford corrections for the gadget}, together with the transformation they induce on the Pauli generators $X$ and $Z$. Knowing these transformations is essential to the actual implementation of the procedure described herein.}.

To prove the universality of PBC, we describe an efficient procedure to transform any such universal quantum circuit, into a corresponding standard PBC consisting of, at most, $t$ measurements of $t$ independent and pairwise commuting Pauli operators on the $t$ qudits in the magic state $\left|T_v\right>^{\otimes t}$. We start by creating an operator list that contains $n$ (dummy) operators of the form $ Z(\boldsymbol{a}_i)\otimes I^{\otimes t} \in \mathcal{P}_{n+t}$, where each $\boldsymbol{a}_i$ is a bit string with $0$ in every bit, except in the $i$th position, so that:
\begin{equation}
    \textsc{list}: Z_1, Z_2, \ldots, Z_n\,.
\end{equation}
At the end of the procedure described below, this list will hold a set of at most $n+t$ independent and compatible Pauli operators. The $t$ operators added to it correspond to those that are measured in the quantum hardware.

Because the quantum circuit obtained by magic-state injection is comprised only of Clifford operations, we can (efficiently) backward-propagate each single-qudit $Z$ measurement until it reaches the beginning of the circuit. We start with each of the $t$ intermediate measurements and then move on to the $m$ final measurements. For the $i$th $Z$ measurement handled in this way, we obtain a Pauli operator $M_i \in \mathcal{P}_{n+t}$ at the beginning of the quantum circuit. The way to handle the measurement of this operator varies according to one of three options:

\vspace{0.25cm}
\noindent\textbf{Case 1.} The Pauli operator $M_i$ is such that the following commutation relationship with an operator $A_j$ in $\textsc{list}$ holds:
\begin{equation}
    M_i A_j = \omega^{\varphi_i} A_j M_i,
\end{equation}
with $\varphi_i \neq 0$. If this is so, it is easy to show that all $p$ possible outcomes of $M_i$ are equiprobable with probability $1/p$. Thus, instead of measuring $M_i$, we can simply get its outcome, $\sigma,$ in a classical way by drawing a random number from the set $\{0,1,\ldots,p-1\}.$ This allows us to remove $M_i$ from the quantum circuit. However, we still need to ensure that the state of the system is changed to the appropriate eigenstate of $M_i$, which we would have gotten had we actually performed the measurement and obtained $\sigma.$ This is done by introducing the following unitary in place of $M_i$:
\begin{equation}
    V_{(M_i,\sigma)}^{(A_j,a)} = \frac{\omega^{a}}{\sqrt{p}} \sum_{k=0}^{p-1} \omega^{-k(\sigma - a)}M_i^k A_j^{-k-1}\,.
    \label{eq: Clifford V(a,sigma)}
\end{equation}
It can be shown that this operator is a Clifford unitary which prepares the appropriate $\sigma$ eigenstate of $M_i$ (cf. Appendix~\ref{app: Proving V(a,s)}). Once this is done, we store the outcome $\sigma$ in a separate list for the outcomes and move on to the subsequent measurement.

\vspace{0.25cm}
\noindent\textbf{Case 2.} Instead, $M_i$ might commute with each operator $A_j$ in $\textsc{list}$ and depend on a certain subset $S$ of those operators so that:
\begin{equation*}
    M_i = \prod_{A_j \in S} A_j\,. 
\end{equation*}
Then its outcome, $\sigma$, can again be obtained classically using the outcome $a_j$ of each $A_j\in S$: $\sigma = \left( \sum_{j:A_j\in S} a_j \right) \text{ mod } p$; $\sigma$ then needs to be added to the list of outcomes.
\begin{figure*}[t]
    \centering
    \begin{tikzpicture}
\node[scale=0.66]{
    \begin{quantikz}[thin lines]
        \lstick[wires=3]{\ket{\Psi_{in}}}  &\gate[wires=3,style={fill=gray!15, rounded corners}]{X^{a}\otimes Z^{b}\otimes X^{c}Z^{d}}  &\qw\rstick[wires=3]{\ket{\Psi_{out}}} \\
                                           &\qw                                                                                         &\qw \\
                                           &\qw                                                                                         &\qw
    \end{quantikz}
$\rightarrow$
    \begin{quantikz}[thin lines]
        \lstick[wires=3]{\ket{\Psi_{in}}}  &\qw\gategroup[wires=4,steps=10,style={rounded corners,fill=gray!15, inner xsep=2pt},background]{{$M = X^{a}\otimes Z^{b}\otimes X^{c}Z^{d}$}}  &\qw                 &\qw                             &\gate{\mathrm{SUM}^a}  &\qw                    &\qw                    &\qw                    &\qw                                &\qw                            &\qw&\qw\rstick[wires=3]{\ket{\Psi_{out}}} \\
                                           &\gate{F^{\dagger}}                                                                                                                             &\qw                 &\qw                             &\qw                    &\gate{\mathrm{SUM}^b}  &\qw                    &\qw                    &\qw                                &\gate{F}                       &\qw&\qw \\
                                           &\gate{X^{2^{-1}}}                                                                                                                              &\gate{Z^{2^{-1}d}}  &\gate{(S^{\dagger})^{c^{-1}d}}  &\qw                    &\qw                    &\gate{\mathrm{SUM}^c}  &\gate{S^{c^{-1}d} }    &\gate{(Z^{\dagger})^{2^{-1}d}}     &\gate{(X^{\dagger})^{2^{-1}}}  &\qw&\qw \\
        \lstick{$\ket{0}_{aux}$}           &\gate{F}                                                                                                                                       &\qw                 &\qw                             &\ctrl{-3}              &\ctrl{-2}              &\ctrl{-1}              &\qw                    &\qw                                &\gate{F^{\dagger}}             &\meter{$\sigma$}
    \end{quantikz}
};
\end{tikzpicture}
    \caption{Scheme for carrying out the measurements in the adaptive PBC for the case of $p$-dimensional qudits (generalization for qudits of the first scheme in Ref.~\cite{PeresGa2022}). The gray box corresponds to the implementation of the Pauli measurement $M = X^{a}\otimes Z^{b}\otimes X^{c}Z^{d}$, as described in the main text. After each measurement, the auxiliary qudit can be reset to $\left| 0 \right>$, and another Pauli operator can be subsequently measured. The factors $2^{-1}$ and $c^{-1}$ represent, respectively, the inverse of $2$ and of $c$ in the finite field $\mathbb{F}_p$; note that $c$ is necessarily different from $0$.}
    \label{fig: Our proposal for performing Pauli measurements on qudits}
\end{figure*}
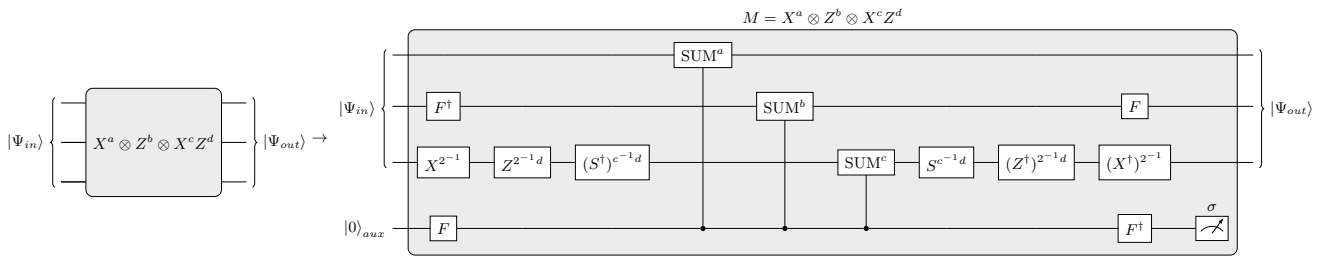

\vspace{0.25cm}
\noindent\textbf{Case 3.} The final option is that $M_i$ commutes with all of the operators in $\textsc{list}$ and is independent from them. This Pauli operator can be written as the tensor product of two operators $M_i^{\prime} \in \mathcal{P}_{n}$ and $M_i^{\prime \prime} \in \mathcal{P}_{t}$: $M_i = M_i^{\prime}\otimes M_i^{\prime \prime}.$ Because $M_i$ commutes with every operator in $\textsc{list}$ , $M_i^{\prime}$ necessarily has the form: $M_i^{\prime} = Z(\boldsymbol{z}),$ with $\boldsymbol{z} \in \mathbb{F}_p^n,$ hence acting trivially on the $n$ stabilizer qudits. This means that we can simply measure the operator $M_i^{\prime \prime}$ on the magic qudits $\left|T_v\right>^{\otimes t}$ to obtain the outcome, $\sigma$, of the quantum measurement. After doing so, the operator $M_i$ must be added to $\textsc{list}$ and $\sigma$ to the list of outcomes. Note that the argument used herein means that any PBC with an input of the form $\left| 0 \right>^{\otimes n} \left| T_v \right>^{\otimes t}$ can be reduced to a smaller PBC involving only the magic register. The former PBC is often called \textit{generalized} PBC, and the latter \textit{standard} PBC~\cite{BSS2016, PeresGa2022}.

\vspace{0.25cm}
This procedure demonstrates how any universal quantum circuit with $t$ non-stabilizer gates can be reduced (in polynomial time) to a sequence of measurements of independent and pairwise commuting Pauli operators on the $t$ qudits in the magic state $\left|T_v\right>^{\otimes t}$. Thus, PBC is a universal model of quantum computation with qudits.

Since there are at most $t$ independent and compatible Pauli operators on $t$ qudits, each PBC involves at most $t$ Pauli measurements. The outcome of the computation is given by the last $m$ values stored in the outcome list.

We reiterate that this generalization holds only for odd-prime qudits. There are several reasons for this. First, as explained in Sec.~\ref{subsec: Universality via MSI}, the claim that the Clifford group supplemented by any non-stabilizer gate is enough for performing any universal quantum computation holds only for odd-prime qudits. Moreover, the magic gates $U_v$ considered have been defined only for that specific case~\cite{HowardVala2012}. Finally, as explained at the end of Appendix~\ref{app: Proving V(a,s)}, the proof that the unitary $V_{(M_i,\sigma_i)}^{(A_j,a)}$, defined in Eq.~\eqref{eq: Clifford V(a,sigma)}, belongs to the Clifford group also holds only for the odd-prime case.

\subsection{Practical implementation}\label{subsec: Practical implementation}

It is natural to wonder how such a sequence of Pauli measurements on $t$ qudits can be implemented in practice. As it turns out, \textit{indirect} measurements are an important topic across several different quantum information and computation applications ~\cite{KnillOS2007,Dobsicek+2007,Dressel+2018,WangHB2019,MitaraiFuji2019}. In Ref.~\cite{PeresGa2022}, such ideas were explored in the context of PBC with qubits; here, we generalize and improve the proposals therein for the case of odd-prime-dimensional qudits. 

\subsubsection{Method 1: one auxiliary qudit and quadratic depth}

Suppose we want to measure a generic Pauli operator of the form given by Eq.~\eqref{eq: generic Pauli operator/mmt} on an arbitrary $t$-qudit state $\left| \Psi_{in} \right>$. To do this, we bring in an auxiliary qudit in state $\left| 0 \right>$ and transform it using the Fourier gate. On each of the $t$ (computational) qudits, the Pauli measurement $M$ may assume one of four possible forms: $I$, $X^{a_i}$, $Z^{b_i}$ or $X^{c_i}Z^{d_i}$, with $a_i,\,b_i,\,c_i,\,d_i \in \mathbb{F}_p$. In the case of the identity, nothing needs to be done with that qudit. When we have $X^{a_i}$ on qudit $i$, we simply apply $a_i$ copies of the \textsc{sum} gate, controlled on the auxiliary qudit and targeting the $i$th computational qudit. In the case where we have $Z^{b_i}$ on qudit $i$, we apply $b_i$ copies of the \textsc{sum} gate, again controlled on the auxiliary qudit and targeting the $i$th computational qudit, conjugated by the Fourier gate on the target qudit.

These first three possibilities are analogous to those in the qubit case, except that here we need to account for the possibility of having different powers of $X$ and $Z$; this is done via multiple applications of the \textsc{sum} gate (as explained above). If $c_i = d_i$, the fourth (and last) possibility can be regarded as the analog of measuring $Y$ in the qubit setting. However, if $c_i \neq d_i$, there is no qubit parallel. In this case, we need to apply the \textsc{sum} gate a number of times $c_i$ equal to the power of $X$; those gates are then conjugated (on the target computational qudit) by the Clifford unitary $\left( X^{\dagger}\right)^{2^{-1}} \left( Z^{\dagger}\right)^{2^{-1}d_i} S^{c_i^{-1}d_i}$, with $2^{-1}$ and $c_i^{-1}$ meaning, respectively, the inverse of $2$ and $c_i$ in the finite field $\mathbb{F}_p$.

Fig.~\ref{fig: Our proposal for performing Pauli measurements on qudits} shows an explicit example of how to measure the three-qudit Pauli operator $M=X^{a}\otimes Z^{b} \otimes X^{c}Z^{d}.$ This example incorporates all non-trivial possible cases. As seen from Eq.~\eqref{eq: generic Pauli operator/mmt}, in general, these measurements may have a global phase factor of the form $\omega^{\lambda},$ so that instead of measuring $M$ we might be interested in measuring $M^{\prime}=\omega^{\lambda} X^{a}\otimes Z^{b} \otimes X^{c}Z^{d}.$ This is not an issue, as the corresponding outcome, $\sigma^{\prime}$, may be inferred by re-interpreting the outcome, $\sigma$, obtained by measuring $M$: $\sigma^{\prime} \coloneqq (\sigma + \lambda) \text{ mod } p.$

Using this scheme, we obtain what the authors in Ref.~\cite{PeresGa2022} call \textit{adaptive PBC circuits} whose features can be easily characterized. Namely, the number of \textsc{sum} gates is upper bounded by $N_{\mathrm{SUM}}^{up.} = (p-1)t^2,$ where $p$ is the dimension of the qudits and $t$ is the number of non-Clifford $U_{v}$ gates in the original quantum circuit. The depth exhibits a similar upper bound of $O\left( (p-1)t^2 \right).$ 

\subsubsection{Method 2: $t$ auxiliary qudits and linear depth}
\begin{figure*}[t]
    \centering
    \begin{tikzpicture}
\node[scale=0.66]{
    \begin{quantikz}[thin lines]
        \lstick[wires=3]{\ket{\Psi_{in}}}  &\gate[wires=3,style={fill=gray!15, rounded corners}]{X^{a}\otimes Z^{b}\otimes X^{c}Z^{d}}  &\qw\rstick[wires=3]{\ket{\Psi_{out}}} \\
                                           &\qw                                                                                         &\qw \\
                                           &\qw                                                                                         &\qw
    \end{quantikz}
$\rightarrow$
    \begin{quantikz}[thin lines]
        \lstick[wires=3]{\ket{\Psi_{in}}}         & \qw\gategroup[wires=6,steps=10,style={rounded corners,fill=gray!15, inner xsep=2pt},background]{{$M = X^{a}\otimes Z^{b}\otimes X^{c}Z^{d}$}}  &\qw                 &\qw                             &\gate{\mathrm{SUM}^a}  &\qw                    &\qw                    &\qw                    &\qw                                &\qw                            &\qw&\qw\rstick[wires=3]{\ket{\Psi_{out}}} \\
                                                  & \gate{F^{\dagger}}                                                                                                                             &\qw                 &\qw                             &\qw                    &\gate{\mathrm{SUM}^b}  &\qw                    &\qw                    &\qw                                &\gate{F}                       &\qw&\qw \\
                                                  & \gate{X^{2^{-1}}}                                                                                                                              &\gate{Z^{2^{-1}d}}  &\gate{(S^{\dagger})^{c^{-1}d}}  &\qw                    &\qw                    &\gate{\mathrm{SUM}^c}  &\gate{S^{c^{-1}d} }    &\gate{(Z^{\dagger})^{2^{-1}d}}     &\gate{(X^{\dagger})^{2^{-1}}}  &\qw&\qw \\
        \lstick[wires=3]{\ket{\mathrm{GHZ}_{3}}}  & \qw                                                                                                                                            &\qw                 &\qw                             &\ctrl{-3}              &\qw                    &\qw                    &\qw                    &\qw                                &\gate{F^{\dagger}}             &\meter{$\sigma_1$} \\
                                                  & \qw                                                                                                                                            &\qw                 &\qw                             &\qw                    &\ctrl{-3}              &\qw                    &\qw                    &\qw                                &\gate{F^{\dagger}}             &\meter{$\sigma_2$} \\
                                                  & \qw                                                                                                                                            &\qw                 &\qw                             &\qw                    &\qw                    &\ctrl{-3}              &\qw                    &\qw                                &\gate{F^{\dagger}}             &\meter{$\sigma_3$}
    \end{quantikz}
};
\end{tikzpicture}
    \caption{Our second scheme for carrying out the measurements in an adaptive PBC with $p$-dimensional qudits. The gray box corresponds to the implementation of the Pauli measurement $M = X^{a}\otimes Z^{b}\otimes X^{c}Z^{d}$, as described in the main text. After each measurement, the auxiliary qudits can be re-initialized in the appropriate GHZ state, and another Pauli operator can be measured. The factors $2^{-1}$ and $c^{-1}$ represent, respectively, the inverse of $2$ and of $c$ in the finite field $\mathbb{F}_p$; note that $c$ is necessarily different from $0$. Unlike what happens with the first method, here the \textsc{sum} gates are parallelized so that each Pauli measurement is performed in constant depth.}\label{fig: 2nd mmt scheme (constant depth)}
\end{figure*}
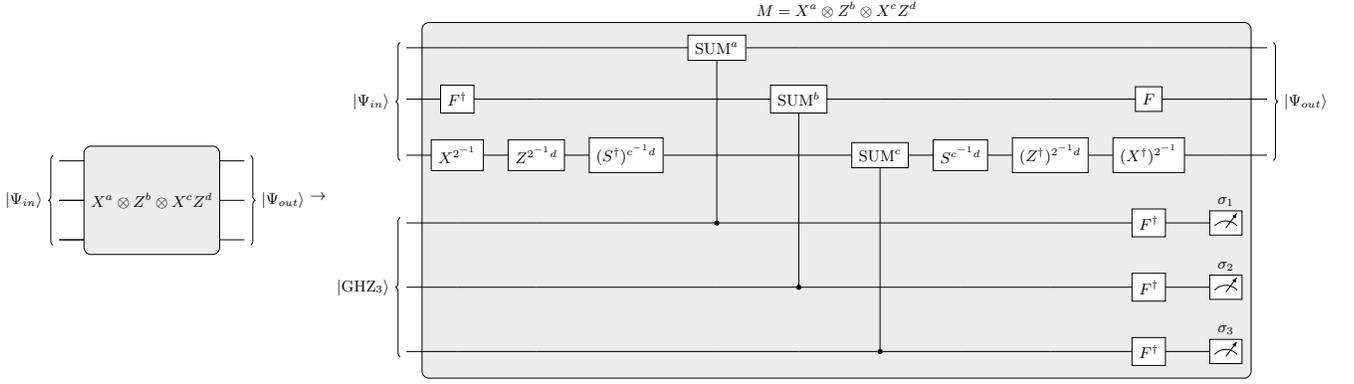

A trade-off is possible between the depth and the width of these adaptive PBC circuits. Specifically, instead of using a single auxiliary qudit, one can use $t$ auxiliary qudits in the (generalized) $t$-qudit Greenberger--Horne--Zeilinger (GHZ) state:
\begin{equation*}
    \left| \mathrm{GHZ}_t \right> = \frac{1}{\sqrt{p}} \sum_{k=0}^{p-1} \left| k \right>^{\otimes t}\,.
    \label{eq: GHZ qudit state}
\end{equation*}
This state can be prepared by a quantum circuit of constant depth, as proved in Appendix~\ref{app: Preparation of an n-qudit GHZ}.

With access to such auxiliary states, each Pauli operator $M_i$ in the PBC can be measured in constant depth. To that end, we apply a procedure analogous to that of the first method, except now the \textsc{sum} gates are controlled on the $i$th auxiliary qudit and target the $i$th computational qudit. After that, each auxiliary qudit is rotated by the inverse of the Fourier gate and measured in the computational basis to yield outcomes $\sigma_1,\sigma_2,\ldots,\sigma_t$. It is  fairly easy to show that doing so leaves the computational qudits in the $\sigma$ eigenstate of $M_i$, with $\sigma$ given by $\sigma = \left( \sum_{i=1}^{t} \sigma_i \right)\text{ mod } p.$ See Fig.~\ref{fig: 2nd mmt scheme (constant depth)} for a depiction of this measurement scheme.

The sequence of \textsc{sum} gates in the measurement procedure described above can now be implemented with a constant depth of at most $(p-1).$ Therefore, although the overall gate count still scales as $O((p-1)t^2),$ the depth of the adaptive PBC circuits constructed in this way has been improved to $O((p-1)t),$ at the expense of an increased circuit width and a small additional classical processing associated to the $\left| \mathrm{GHZ}_t \right>$ preparation.

\subsubsection{Further improvements}\label{subsubsec: Further improvements}

The two-qudit gate count and depth previously obtained can actually be further improved through an observation that has no qubit counterpart. Consider a generic $t$-qudit Pauli operator $M$ of the form given by Eq.~\eqref{eq: generic Pauli operator/mmt} whose measurement yields an outcome $\sigma$. Using the constructions previously discussed, performing this measurement involves a number of \textsc{sum} gates given by:
\begin{equation}
    N_{\mathrm{SUM}} = \sum_{j=1}^{t} \left( x_j + z_j \delta_{x_j,0}\right)\,.
\end{equation}
However, instead of being limited to measuring only this operator, one can measure any of $(p - 1)$ Pauli operators of the form:
\begin{equation}
    M^{\prime} = \omega^{k \lambda} X(k\boldsymbol{x}) Z(k\boldsymbol{z})\,,
    \label{eq: Any allowed Mprime}
\end{equation}
where $k\in \mathbb{F}_p \backslash \{0\}$ and $k\boldsymbol{v} = \left( kv_1,\,\ldots,\,kv_t\right) \in \mathbb{F}_p^{t}.$ The corresponding measurement outcome, $\sigma^{\prime}$, must be re-interpreted so that:
\begin{equation}
    \sigma \coloneqq \sigma^{\prime} k^{-1} + \boldsymbol{x}\cdot \boldsymbol{z} (k-1)/2
\end{equation}
For proof of this result, see Appendix~\ref{app: Improvement of the practical mmts}.

The associated gate count is then given by:
\begin{equation}
    N_{\mathrm{SUM}}^{\prime} = \sum_{j=1}^{t} \left( kx_j + kz_j \delta_{x_j,0}\right)\,.
\end{equation}
Note that while $kx_j,kz_j \in \mathbb{F}_p$, all the additions in the equation above are performed over $\mathbb{Z}$. Finding the best number of two-qudit gates reduces to finding the value of $k$ for which the equation above yields the minimum value. This is something that can be done efficiently in the assisting classical computer, before calling on the quantum processing unit (QPU) to perform the quantum measurement.

Albeit this reduction corresponds only to a constant saving in each measured Pauli operator, this is still of practical relevance for implementations in earlier prototypes. The same can be said about the upper bound on the depth and two-qudit gate count of the entire PBC circuits: $O\left( pt^2 / 2\right)$ with Method 1, which corresponds to a reduction by a (constant) factor of 2. Similarly, if using Method 2, this observation leads to analogous (constant) improvements, namely, adaptive PBC circuits with $O\left( pt / 2\right)$ depth.

\subsection{Hybrid computation}\label{subsec: Hybrid computation}

Suppose we want to estimate the probability of getting the outcome $y=0$ as the output of a standard PBC on $t$ qudits. If the available QPU has a number of qudits greater than $t$, this problem is straightforward. We can run the computation $N$ times and count the number of occurrences of the outcome $y=0$, $N_0$, estimating the desired probability as $q_0 = N_0/N$. If we wish our result to be accurate up to a precision $\epsilon$, then $N=O\left(\epsilon^{-2} \right)$ samples suffice.

Instead, consider that the QPU has access to only $t-k$ qudits. In this case, the PBC is too large to run directly on the available hardware. As it turns out, this problem has a simple solution in the PBC framework~\cite{BSS2016, PeresGa2022}. We start by noting that the input state for the first $k$ qudits can be decomposed into a linear superposition of stabilizer states:
\begin{equation}
    \left| T_v \right> \left< T_v \right|^{\otimes k} = \sum_{j=0} c_j \left| \phi_j \right> \left< \phi_j \right|\,,
    \label{eq: magic state decomposition}
\end{equation}
where the coefficients $c_j$ are real. Each stabilizer state $\left| \phi_j \right>$ can be obtained from the ground state $\left| 0 \right>^{\otimes k}$ via the application of a $k$-qudit Clifford unitary $C_j$ which can be determined efficiently. Thus, we can write the input state of the PBC as 
\begin{equation*}
    \left| T_v \right> \left< T_v \right|^{\otimes t} = \sum_{j=0} c_j C_j \left| 0 \right> \left< 0 \right|^{\otimes k} C_j^{\dagger} \otimes \left| T_v \right> \left< T_v \right|^{\otimes (t-k)}\,.
\end{equation*}

To estimate the desired probability, $q_0$, we can sample the state of the first $k$ qudits from the probability distribution given by $\bar{q}_{j} = \left| c_j \right| / \lVert \boldsymbol{c} \rVert_1 ,$ where $\lVert \boldsymbol{c} \rVert_1 = \sum_j \left| c_j \right|$ denotes the $\ell_1$ norm of the coefficients of the decomposition in Eq.~\eqref{eq: magic state decomposition}~\cite{HowCamp2017, Cai+2022}. Once we have done so, we are left with a PBC with input state $C_j \left| 0 \right>^{\otimes k} \otimes \left| T_v \right>^{\otimes (t-k)}$. Such a (generalized) PBC on $t$ qudits can be reduced to a (standard) PBC acting only on the magic register, i.e., a PBC on $(t-k)$ qudits (recall Sec.~\ref{subsec: Proof of universality}). We can then run that computation to get an outcome $m \in \mathbb{F}_p$.

It can be shown that the random variable
\begin{equation*}
    \eta = \frac{1}{p} + \frac{1}{p} \lVert \boldsymbol{c} \rVert_1 \mathrm{sign}(c_j) \sum_{\mu=1}^{p-1} \omega^{\mu m}
\end{equation*}
is an unbiased estimator of the probability, $q_0$, where $\mathrm{sign}(c_j)$ denotes the signal of the appropriate coefficient $c_j$ and $m$ is the outcome obtained from running the corresponding smaller PBC. Thus, repeating the procedure described above a suitable number of times allows us to gather samples to estimate $q_0$. Importantly, the random variable $\eta$ is real and bounded in the interval $\left[ \frac{1}{p}\left(1 - \lVert \boldsymbol{c} \rVert_1 ^2 (p-1) \right) ,  \frac{1}{p}\left(1 + \lVert \boldsymbol{c} \rVert_1 ^2 (p-1) \right) \right]$. Therefore, to get an estimate of $q_0$ to within $\epsilon$ of its actual value with probability at least $(1-q_{fail})$, Hoeffding's inequality informs us that the required number of samples is $O\left( \lVert \boldsymbol{c} \rVert_1 ^2 \epsilon^{-2} \ln\left( 2/q_{fail} \right) \right)$~\cite{HowCamp2017, Cai+2022}. Clearly, in this framework, finding the optimal sampling complexity corresponds to minimizing the $\ell_1$ norm of $\boldsymbol{c}$. This corresponds to calculating a magic monotone known as robustness of magic (RoM)~\cite{HowCamp2017}. The RoM of an arbitrary $n$-qudit state $\rho$ is given by:
\begin{equation}
    \mathcal{R}\left( \rho \right) \coloneqq \min_{\boldsymbol{c}} \Biggl\{ \sum_j \left| c_j \right|;\,\rho = \sum_i c_i \sigma_i \Biggr\}\,,
\end{equation}
where $\sigma_i$ are $n$-qudit stabilizer states.

For the qubit case, the results from Ref.~\cite{Heinrich2019} allow us to upper bound the sampling complexity of this hybrid quantum-classical scenario by $O\left( 2^{0.7374 k} \epsilon^{-2}\right),$ where $k$ is the desired number of virtual qubits.

Analogously to the works in Refs.~\cite{HowCamp2017, Heinrich2019}, we can compute the RoM of tensor products of the magic states $\left| T_v\right>$ for $p$-dimensional systems. This will give us an upper bound for the sampling complexity of hybrid computation in the qudit setting. Here, we consider only qutrits ($p=3$) and ququints ($p=5$). The results obtained for the RoM of multiple copies of such magic states, $\left| T_v \right>$, of higher-dimensional systems can be seen in Table~\ref{tab: RoM results for qutrits and ququints}, and lead to a sampling complexity upper bounded by $O\left( 3^{ 1.0848 k} \epsilon^{-2}\right)$ and $O\left( 5^{ 1.4022 k} \epsilon^{-2}\right)$, respectively, for $p=3$ and $p=5$. We note that the numerical coefficients in the exponents increase with $p$. This seems to suggest that the cost of simulating virtual qudits increases with the system's dimension by an amount that is larger than what can be explained by the dimensionality increase alone (encoded in the base of the exponential).
\begin{table}[]
    \vspace{-0.21cm}\caption{Robustness of magic, $\mathcal{R},$ for the qutrit (top) and ququint (bottom) magic states $\left| T_v \right>$ defined, respectively, by $v=(0,1,8)$ and $v=(0,3,4,2,1).$}
    \centering
    \begin{tabular}{>{\centering\arraybackslash}p{0.6cm} >{\centering\arraybackslash}p{2cm} >{\centering\arraybackslash}p{2.1cm}}
        \hline\hline
        $n$ & $\mathcal{R}\left(\left| T_v \right>^{\otimes n}\right)$ & $\mathcal{R}\left(\left| T_v \right>^{\otimes n}\right)^{2/n}$\\ \hline
        1 & 1.94098 & 3.76741 \\
        2 & 3.44194 & 3.44194 \\
        3 & 5.97505 & 3.29277 \\
        \hline\hline
    \end{tabular}
    
    \vspace{0.3cm}
    \begin{tabular}{>{\centering\arraybackslash}p{0.6cm} >{\centering\arraybackslash}p{2cm} >{\centering\arraybackslash}p{2.1cm}}
        \hline\hline
        $n$ & $\mathcal{R}\left(\left| T_v \right>^{\otimes n}\right)$ & $\mathcal{R}\left(\left| T_v \right>^{\otimes n}\right)^{2/n}$\\ \hline
        1 & 3.43607 & 11.80656 \\
        2 & 9.55197 &  9.55197 \\
        \hline\hline
    \end{tabular}
    \label{tab: RoM results for qutrits and ququints}
\end{table}

To strengthen the latter conjecture, we can compute lower bounds on the cost of this particular scheme for hybrid computation with systems of different dimensions. This is done by lower-bounding the RoM of magic states. In recent work by Leone \textit{et al.}~\cite{LeoneOH2022}, the authors define a measure of the magic of a (pure) state $\left| \psi \right>$ called \textit{stabilizer $\alpha$-Rényi entropy}. In their work, they are focused on qubit systems. Here, we make a straightforward generalization of this measure which applies to systems of arbitrary dimension $p$.

Let us denote by $\tilde{\mathcal{P}}_{n}$ the set of $n$-qudit Pauli operators with global phase equal to $1$. Clearly, this set contains $p^{2n}$ operators. Next, we define $\Xi_{P} \left( \left| \psi \right> \right) \coloneqq p^{-n} \left| \left< \psi \left| P \right| \psi \right> \right|^2$, where $\left| \psi \right>$ is an $n$-qudit pure state \footnote{Note the difference in the definition of $\Xi_{P}\left( \left| \psi \right> \right)$ with respect to the one in~\cite{LeoneOH2022}. There, the authors write only the square of the expectation value of $P$, whilst here we write the squared modulus of $\left< \psi \left| P \right| \psi \right>$. This is because, unlike what happens with qubits, the qudit Pauli operators are not Hermitian; therefore, their eigenvalues are not guaranteed to be real.}. With this definition, the values of $\Xi_{P} \left( \left| \psi \right> \right)$ are guaranteed to be real values such that $\Xi_{P} \left( \left| \psi \right> \right) \geq 0.$ Moreover, it is not hard to show that $\sum_{P\in \tilde{\mathcal{P}}_n} \Xi_{P} \left( \left| \psi \right> \right) = 1.$ Hence, $\{\Xi_{P} \left( \left| \psi \right> \right)\}$ can be regarded as a probability distribution (in perfect analogy to Ref.~\cite{LeoneOH2022}). The $\alpha$-Rényi entropies, $M_{(\alpha,p)}$, associated with $\{\Xi_{P} \left( \left| \psi \right> \right)\}$, are defined as:
\begin{equation}
    M_{(\alpha, p)} \left( \left| \psi \right> \right) \coloneqq \frac{1}{1-\alpha} \log_p \sum_{P\in \tilde{\mathcal{P}}_n} \Xi_{P}^{\alpha} \left( \left| \psi \right> \right) - n\,.
    \label{eq: stabilizer Renyi entropy}
\end{equation}
This definition preserves all the relevant features of the measure defined in Ref.~\cite{LeoneOH2022}; namely, faithfulness, stability under free operations, and additivity. For explicit proofs see Appendix~\ref{app: Gen Stab Renyi Entropies}.

Next, in Ref.~\cite{LeoneOH2022}, the authors note that, for the qubit case, the \nicefrac{1}{2}-Rényi entropy is related to a witness of non-stabilizerness known as the stabilizer-norm (or simply st-norm)~\cite{Campbell2011} defined as~\cite{Campbell2011, HowCamp2017}:
\begin{equation*}
    \mathcal{D} \left( \rho \right) \coloneqq \frac{1}{2^n} \sum_{P\in \tilde{\mathcal{P}}_n} \left| \mathrm{Tr} \left( P \rho \right) \right|\,, 
\end{equation*}
where $\rho$ is some $n$-qubit density matrix. 
Since $\mathcal{D} \left( \rho \right) \leq \mathcal{R} \left( \rho \right)$~\cite{HowCamp2017} and, for a pure state $\rho = \left| \psi \right> \left< \psi \right|$, $M_{\nicefrac{1}{2}} ( \left| \psi \right> ) = 2\log_2 \mathcal{D} \left( \left| \psi \right> \right),$ it is straightforward that $M_{\nicefrac{1}{2}} \left( \left| \psi \right> \right) \leq 2 \log_2 \mathcal{R}\left( \left| \psi \right> \right).$ This can be easily generalized for $p$-dimensional systems: $M_{(\nicefrac{1}{2},p)} \left( \left| \psi \right> \right) \leq 2 \log_p \mathcal{R}\left( \left| \psi \right> \right),$ which allows us to write the following lower bound for the RoM of any (pure) state $\left| \psi \right>$:
\begin{equation}
    \mathcal{R}\left( \left| \psi \right> \right)^{2} \geq p^{M_{\left(\nicefrac{1}{2},p\right)} \left( \left| \psi \right> \right)}\,.
    \label{eq: Hybrid computation cost lower bound}
\end{equation}

Thus, computing the $\nicefrac{1}{2}$-Rényi entropy of $k$ copies of magic states allows us to lower-bound the cost of hybrid computation in the PBC framework. The result in the Supplementary Material of Ref.~\cite{LeoneOH2022} allows us to write that $\mathcal{R}\left( \left| T \right> ^{\otimes k}\right)^{2} \geq 2^{2\{ \log_2 \left[ \nicefrac{(\sqrt{2} + 1)}{2} \right] \}k},$ which means a lower bound on hybrid computation with qubits given by $\Omega \left( 2^{0.5431 k} \epsilon^{-2} \right).$

Computing the $\nicefrac{1}{2}$-Rényi entropy of $k$ copies of the qutrit and ququint magic states leads to the lower bounds of $\Omega\left( 3^{0.7236 k} \epsilon^{-2} \right)$ and $\Omega\left( 5^{0.8544 k} \epsilon^{-2} \right)$ for hybrid computation with qutrits and ququints, respectively. For detailed calculations see Appendix~\ref{app: Gen Stab Renyi Entropies}.

\section{Conclusions}\label{sec:Conclusions}

Albeit most work in quantum computing has been developed for two-level systems, it is conceivable that the future lies in higher-dimensional systems. Motivated by this, we proved the universality of the Pauli-based model of quantum computation with odd-prime qudits. This was done by showing that any universally general quantum circuit with $t$ magic gates $U_v$ can be transformed into a Pauli-based computation on $t$ qudits and (at most) $t$ measurements of independent and pairwise commuting Pauli operators on $t$ qudits.

Additionally, we explained how a sequence of $t$-qudit Pauli measurements can be implemented in practice by describing how it can be translated into adaptive PBC circuits. We presented two different proposals for this translation. In the first, a single auxiliary qudit is needed, so that the final adaptive circuits have $t+1$ qudits and require $O\left( pt^2/2 \right)$ \textsc{sum} gates and depth. The second proposal explores a trade-off between depth and width. It requires $t$ auxiliary qudits but is able to bring the depth down to $O\left( pt /2 \right)\,.$

Next, we analyzed the robustness of magic of tensor products of qutrit and ququints magic states, $\left| T_v \right>.$ The results allowed us to upper bound the sample complexity associated with simulating $k$ virtual qutrits and ququint by $O\left( 3^{ 1.0848 k} \epsilon^{-2}\right)$ and $O\left( 5^{ 1.4022 k} \epsilon^{-2}\right)$, respectively. These results suggest that the cost of hybrid computation increases with the qudit's dimension by an amount that goes beyond what is justifiable by the increase of $p$ alone. The validity of this hypothesis is further supported by the lower bounds computed for $p=2$, $p=3$, and $p=5$, where the same behavior is observed. Specifically, the lower bounds of the sampling complexity of hybrid computation with qubits, qutrits, and ququints are, respectively, $\Omega \left( 2^{0.5431 k} \epsilon^{-2} \right),$ $\Omega\left( 3^{0.7236 k} \epsilon^{-2} \right)$, and $\Omega\left( 5^{0.8544 k} \epsilon^{-2} \right)$.

We would like to point out that both the lower and upper bounds of hybrid computation presented in this paper hold for the particular scheme described. As explained, this was based on a Monte-Carlo-type algorithm whose (exponential) sampling complexity is determined by the square of the robustness of magic. However, other schemes may exist that enable an improvement of this cost. Finding the optimal approach to this problem of extending the quantum memory through the simulation of virtual qudits on an assisting classical (super)computer could be a relevant prospective line of research. Additionally, it would be interesting to confirm our conjecture that hybrid computation becomes increasingly challenging as $p$ increases and understand its underlying cause.

To conclude, we note that while this work broadens the current realm of applicability of PBC, many things still remain unexplored. Notably, it would be interesting to start expressing quantum algorithms in the PBC framework. This is a non-trivial task, but doing so would allow us to compare the quantum resources needed within PBC to those required within the quantum circuit model or the 1WM. Additionally, it could potentially highlight specific scenarios where the use of higher-dimensional qudits brings significant advantages over qubits. It is also important to start connecting the myriad of works employing PBC in very different contexts, using them to gain a better understanding of this computational model, its peculiar features, advantages, and pitfalls. Finding (new) ways of optimizing PBC for near- and intermediate-term applications, but also for fault-tolerant quantum computing, is another relevant future line of research deserving of further attention.

\section*{Acknowledgements}
FCRP would like to thank Adán Cabello and Michael Zurel for showing interest in this problem, Michael de Oliveira and JMB Lopes dos Santos for fruitful exchanges of ideas, and Mark Howard for helpful insights and discussions. The author also appreciates comments on the text made by Mark Howard and Ernesto Galvão.
FCRP is supported by the Portuguese institution FCT – Funda\c{c}\~{a}o para a Ci\^{e}ncia e a Tecnologia via the Ph.D. Research Scholarship 2020.07245.BD.
This work was supported by the Digital Horizon Europe project FoQaCiA, GA no.101070558, funded by the European Union, NSERC (Canada), and UKRI (U.K.).

\appendix
\section{Clifford corrections for the magic-state injection gadget}\label{app: Clifford corrections for the gadget}

We start by considering the qutrit case ($p=3$) wherein the Clifford correction, $C_{\sigma} = U_{v} \left( X^{\dagger} \right)^{\sigma} U_{v}^{\dagger}$, associated with each possible outcome $\sigma \in \{ 0,\,1,\,2 \}$, is:
\begin{equation}
    C_{\sigma} = \sum_{k=0}^2 \zeta^{v_k-v_{k+\sigma}} \left| k \right> \left< k + \sigma \right|.
    \label{eq: Clifford correction p=3}
\end{equation}
The explicit form of this correction is not as important to us as the transformation it induces, by conjugation, on the generators of the Pauli group. This can be computed explicitly in a fairly straightforward way. We start with the transformation of $Z$, which can be obtained using Eqs.~\eqref{eq: generalized X and Z} and~\eqref{eq: Clifford correction p=3}:
\begin{equation}
    Z \xrightarrow{C_{\sigma}} C_{\sigma}ZC_{\sigma}^{\dagger} =  \omega^{\sigma}Z.
    \label{eq: Transformation of Z under gadget corrections (p=3)}
\end{equation}
Next, we compute the transformation of $X$. Again, making direct use of Eqs.~\eqref{eq: generalized X and Z} and~\eqref{eq: Clifford correction p=3}, we obtain:
\begin{equation*}
    X \xrightarrow{C_{\sigma}} C_{\sigma}XC_{\sigma}^{\dagger} =  \sum_{k=0}^{2} \zeta ^{v_{k+1}-v_k + v_{k+\sigma} -v_{k+\sigma+1}} \left| k + 1 \right> \left< k \right|\,.
\end{equation*}
Following calculations analogous to the ones illustrated in Ref.~\cite{HowardVala2012}, we can obtain an explicit recurrence relation for $v_{k+1} - v_{k}$: $v_{k+1} - v_{k} = 3 \epsilon^{\prime} + 2\gamma^{\prime} + 6z^{\prime} + 3k(z^{\prime} + 2k\gamma^{\prime}),$ for $p=3$. Using this, it is straightforward to show that:
\begin{equation*}
    v_{k+1}-v_k + v_{k+\sigma} -v_{k+\sigma+1} = 3\sigma \left( \gamma^{\prime} \sigma + 2 z^{\prime}\right) + 6\gamma^{\prime}\sigma k\,.
\end{equation*}
Then, it is simple to show that:
\begin{equation}
    X \xrightarrow{C_{\sigma}} C_{\sigma}XC_{\sigma}^{\dagger} =  \omega^{\sigma \left( \gamma^{\prime} \sigma + 2 z^{\prime}\right)}XZ^{2\gamma^{\prime}\sigma}\,.
    \label{eq: Transformation of X under gadget corrections (p=3)}
\end{equation}

Next, we take $p>3$; in this case, the Clifford correction can be written as:
\begin{equation}
    C_{\sigma} = \sum_{k=0}^{p-1} \omega^{v_k-v_{k+\sigma}} \left| k \right> \left< k + \sigma \right|.
\end{equation}
Following a similar approach as above, it can be shown that:
\begin{equation}
    X \xrightarrow{C_{\sigma}} \omega^{- \sigma \left(  2^{-1} \gamma^{\prime} \sigma + z^{\prime} \right)}XZ^{-\gamma^{\prime}\sigma} \text{ and } Z \xrightarrow{C_{\sigma}} \omega^{\sigma}Z \,,
    \label{eq: Transformation of Paulis under gadget corrections (p>3)}
\end{equation}
where $2^{-1}$ is understood to be the inverse of $2$ modulo $p$.

Note that knowing the transformations given by Eqs.~\eqref{eq: Transformation of Z under gadget corrections (p=3)},~\eqref{eq: Transformation of X under gadget corrections (p=3)}, and~\eqref{eq: Transformation of Paulis under gadget corrections (p>3)} is sufficient to efficiently (backward) propagate any Pauli operator through the gadget corrections $C_{\sigma}$, as required for the procedure described in Sec.~\ref{subsec: Proof of universality}.

\section{Proving the validity of $V_{(M,\sigma)}^{(A,a)}$ in Eq.~\eqref{eq: Clifford V(a,sigma)}}\label{app: Proving V(a,s)}

Here, we consider the situation where the $i$th Pauli operator $M$ does not commute with an operator $A$ from the list of dummy $Z$ operators and previously measured, independent and pairwise commuting operators. The outcome associated with the measurement of $A$ is denoted $a$. The commutation relationship between $A$ and $M$ is $MA=\omega^{\varphi}AM,$ where $\varphi \neq 0\,.$

The state of the system before the measurement of $M$ is $\left| \psi_i \right>$. Because the PBC procedure guarantees that all measured operators are compatible, $\left| \psi_i \right>$ is necessarily an eigenstate of $A$ with eigenvalue $a$. That is, $A^k \left| \psi_i \right> = \omega^{ka} \left| \psi_i \right>\,.$

If we were to perform the measurement of $M$, with an outcome $\sigma$, the system would evolve into the state $\left| \psi_{i+1}\right> \propto \hat{P}_{(M,\sigma)} \left| \psi_i \right>,$ where $\hat{P}_{(M,\sigma)}$ denotes the projector associated with the Pauli operator $M$ and corresponding outcome $\sigma$, so that:
\begin{equation*}
    \hat{P}_{(M,\sigma)} = \frac{1}{p} \sum_{k=0}^{p-1} \omega^{-k\sigma} M^k \,.
\end{equation*}
We can append to this operator any power of $\omega^{-a} A$, without changing its action on the state $\left| \psi_i \right>.$ Thus, we can write Eq.~\eqref{eq: Clifford V(a,sigma)}, and it is clear that the action of $V_{(M,\sigma)}^{(A,a)}$ onto $\left| \psi_i \right>$ yields the proper state $\left| \psi_{i+1} \right>$. To lighten the notation, from now on, we will refer to this operator simply as $V$, leaving the dependence on the operators $M$ and $A$, and their corresponding outcomes $\sigma$ and $a$, implicit. 

We now need only prove that (i) $V$ is unitary and (ii) $V$ is a Clifford.

To prove unitarity, we construct $V^{\dagger}$ and show that $V^{\dagger} V = V V^{\dagger} = I\,.$ Straightforward calculation leads to:
\begin{equation*}
    V^{\dagger} V = \frac{1}{p} \sum_{j,k=0}^{p-1} \omega^{(k-j)\left[ a - \sigma - \varphi(k+1) \right]} A^{j - k} M ^{k-j}.
\end{equation*}
Carrying out the relabeling $l \coloneqq k - j$, the equation above can be rewritten as \footnote{Note that, in principle, the start and end points of the summation are shifted by this relabeling. However, in practice, because arithmetic is carried out modulo $p$, the summation is performed exactly over the same range of values, and we can leave the start and end points unchanged.}:
\begin{equation*}
\begin{split}
    V^{\dagger} V & = \frac{1}{p} \sum_{l=0}^{p-1} \omega^{l(a-\sigma - \varphi)} A^{-l} M^{l} \sum_{k=0}^{p-1} \left[ \omega^{-l\varphi}\right]^k \\
    & = \frac{1}{p} \sum_{l=0}^{p-1} \omega^{l(a-\sigma - \varphi)} A^{-l} M^{l} \left(p\, \delta_{l,0} \right) = I\,.
\end{split}
\end{equation*}
An analogous calculation can be performed for $V V^{\dagger}.$
\begin{figure*}
    \centering
\begin{tikzpicture} \node[scale=0.8]{     \begin{quantikz}[thin lines]         \lstick{\ket{0}}\gategroup[wires=2,steps=3,style={dashed,rounded corners,fill=blue!20, inner xsep=10pt, xshift=-6pt, inner ysep=1pt},background]{$\left| \mathrm{GHZ}_2 \right>$}  &\gate{F}  &\ctrl{1}             &\qw                  &\qw            &\qw                      &\qw                    &\qw  &\qw                 &\qw                  &\qw\rstick[wires=8]{\ket{\mathrm{GHZ}_8}}\\         \lstick{\ket{0}}                                                                                                                                                          &\qw       &\gate{\mathrm{SUM}}  &\ctrl{1}             &\qw            &\qw                      &\ctrl{2}               &\qw  &\qw                 &\ctrl{1}             &\qw\\         \lstick{\ket{0}}\gategroup[wires=2,steps=3,style={dashed,rounded corners,fill=blue!20, inner xsep=10pt, xshift=-6pt, inner ysep=1pt},background]{}                        &\gate{F}  &\ctrl{1}             &\gate{\mathrm{SUM}}  &\meter{$m_1$}  &\cwbend{5}               &                       &     &\lstick{$\ket{0}$}  &\gate{\mathrm{SUM}}  &\qw\\         \lstick{\ket{0}}                                                                                                                                                          &\qw       &\gate{\mathrm{SUM}}  &\ctrl{1}             &\qw            &\gate{X^{-m_1}}          &\gate{\mathrm{SUM}^2}  &\qw  &\qw                 &\ctrl{1}             &\qw\\         \lstick{\ket{0}}\gategroup[wires=2,steps=3,style={dashed,rounded corners,fill=blue!20, inner xsep=10pt, xshift=-6pt, inner ysep=1pt},background]{}                        &\gate{F}  &\ctrl{1}             &\gate{\mathrm{SUM}}  &\meter{$m_2$}  &\cwbend{3}               &                       &     &\lstick{$\ket{0}$}  &\gate{\mathrm{SUM}}  &\qw\\         \lstick{\ket{0}}                                                                                                                                                          &\qw       &\gate{\mathrm{SUM}}  &\ctrl{1}             &\qw            &\gate{X^{-m_2+m_1}}      &\ctrl{2}               &\qw  &\qw                 &\ctrl{1}             &\qw\\         \lstick{\ket{0}}\gategroup[wires=2,steps=3,style={dashed,rounded corners,fill=blue!20, inner xsep=10pt, xshift=-6pt, inner ysep=1pt},background]{}                        &\gate{F}  &\ctrl{1}             &\gate{\mathrm{SUM}}  &\meter{$m_3$}  &\cwbend{1}               &                       &     &\lstick{$\ket{0}$}  &\gate{\mathrm{SUM}}  &\qw\\         \lstick{\ket{0}}                                                                                                                                                          &\qw       &\gate{\mathrm{SUM}}  &\qw                  &\qw            &\gate{X^{-m_3+m_2-m_1}}  &\gate{\mathrm{SUM}^2}  &\qw  &\qw                 &\qw                  &\qw     \end{quantikz} }; \end{tikzpicture}
    \caption{Constant-depth quantum circuit for creating an eight-qudit GHZ state, assisted by measurements, adaptivity, and qudit reset.}
    \label{fig: Generation of t-qudit GHZ states}
\end{figure*}
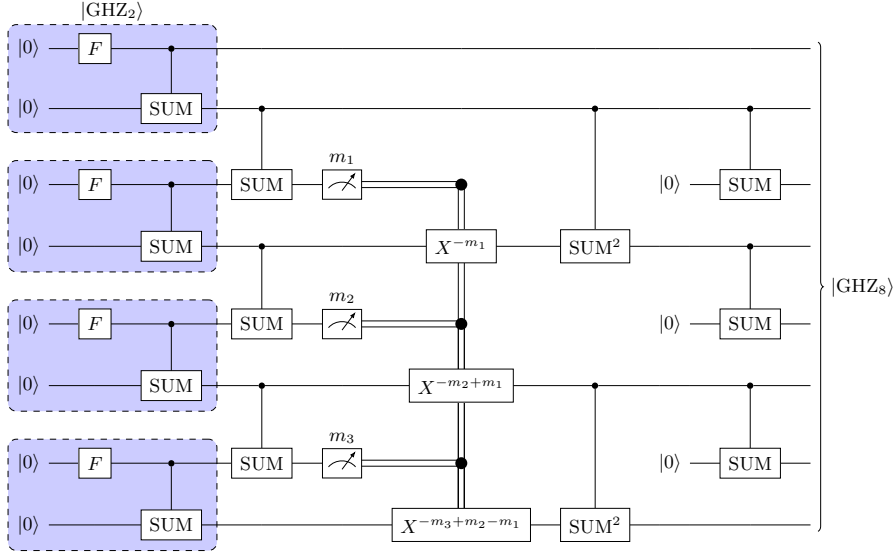

Finally, we show that $V$ is not only a unitary, but it is also a Clifford. To that end, we need to prove that it maps Pauli operators into Pauli operators under conjugation. Consider a (perfectly generic) Pauli operator $R$ so that: $MR = \omega^{\alpha}RM$ and $AR = \omega^{\beta}RA.$ Then,
\begin{equation*}
\begin{split}
    VRV^{\dagger} & = \frac{1}{p} \sum_{j,k=0}^{p-1} \omega^{(k-j)(a-\sigma)} M^{k}A^{-k-1}RA^{j+1}M^{-j} \\
    & = \frac{1}{p} R \sum_{j,k=0}^{p-1} \omega^{(k-j)(a-\sigma)- \beta} \omega^{k(\alpha - \beta + \varphi(j-k))} A^{j-k} M^{k-j}\,.
\end{split}
\end{equation*}
As before, we can perform a relabeling of the indices so that $l \coloneqq k - j$ which leads to the following equation:
\begin{equation}
\begin{split}
    VRV^{\dagger} & = \frac{1}{p} R \sum_{l=0}^{p-1} \omega^{l(a-\sigma)- \beta} A^{-l} M^{l} \sum_{k=0}^{p-1} \left[ \omega^{\alpha - \beta - \varphi l} \right]^k \\
    & = \frac{1}{p} R \sum_{l=0}^{p-1} \omega^{l(a-\sigma)- \beta} A^{-l} M^{l} \left( p \delta_{l,\,\varphi^{-1}(\alpha - \beta)} \right) \\
    & = \omega^{\varphi^{-1}(\alpha - \beta)(a - \sigma) - \beta} R A^{-\varphi^{-1}(\alpha - \beta) }M^{\varphi^{-1}(\alpha - \beta)}\,,
    \label{eq: Proof that V is a Clifford}
\end{split}
\end{equation}
where $\varphi^{-1}$ is understood to be the inverse of $\varphi$ modulo $p$. Because the result in this equation is the product of Pauli operators, it is necessarily a Pauli operator. Thus, $V$ maps any Pauli operator into a Pauli operator.

Knowing the transformation given by Eq.~\eqref{eq: Proof that V is a Clifford} is all that is needed in order to efficiently backward-propagate any Pauli operator through any possible $V$ operator that may arise from the PBC procedure. We note that this holds only when $p$ is an odd-prime integer in which case $\mathbb{F}_p$ is a field and every element $\varphi$ will have a multiplicative inverse $\varphi^{-1}.$ If $p$ is not an odd-prime number then the set $\{ 0, 1, \dots, p-1 \}$ is a ring (rather than a field), and not every element has a corresponding multiplicative inverse.

\section{Preparation of $n$-qudit GHZ states in constant depth}\label{app: Preparation of an n-qudit GHZ}

Here, we show that an $n$-qudit GHZ state $\left| \mathrm{GHZ}_n \right> = \frac{1}{\sqrt{p}} \sum_{k=0}^{p-1} \left|k\right>^{\otimes n}$ can be prepared via a constant depth quantum circuit, as depicted in Fig.~\ref{fig: Generation of t-qudit GHZ states}. This is a generalization for qudits of the construction in Fig. 1 of~\cite{Quek+2022}.

The procedure starts with the preparation of $n/2$ GHZ states on two qudits, as indicated by the blue boxes in Fig.~\ref{fig: Generation of t-qudit GHZ states}. It is straightforward to show that
\begin{equation*}
    \left| \mathrm{GHZ}_2 \right> = \mathrm{SUM}_{1,2}\left( F \otimes I \right) \left| 0,0 \right> = \frac{1}{\sqrt{p}} \sum_{x\in \mathbb{F}_p} \left| x,\,x \right>\,.
\end{equation*}

Having prepared $n/2$ copies of this GHZ state one applies a sequence of $(n/2 - 1)$ \textsc{sum} gates between the second qudit of each pair, and the first qudit of the following one so that the state of the system is transformed into:
\begin{equation*}
\begin{split}
    \left| \psi_1 \right> & = \prod_{j=1}^{n/2 - 1} \mathrm{SUM}_{2j, 2j+1} \left| \mathrm{GHZ}_2 \right>^{\otimes n/2} \\
    & = \frac{1}{\sqrt{p}^{n/2}} \prod_{j=1}^{n/2 - 1} \mathrm{SUM}_{2j, 2j+1} \bigotimes_{k=1}^{n/2} \sum_{x_k \in \mathbb{F}_p} \left| x_k,\,x_k \right> \\
    & = \frac{1}{\sqrt{p}^{n/2}} \sum_{x_1 \in \mathbb{F}_p} \left| x_1,\,x_1 \right> \bigotimes_{k=2}^{n/2} \sum_{x_k \in \mathbb{F}_p} \left| x_k + x_{k-1},\,x_k \right>
\end{split}
\end{equation*}

The next step involves measuring every odd qudit (except the first). Denoting by $m_1,m_2,\ldots,m_{n/2-1}$ each of the measurement outcomes, it is clear from the equation above that $m_{k-1} = x_{k} + x_{k-1}\,.$ This allows one to write a recurrence relationship for $x_k$ so that: $x_k = m_{k-1} - x_{k-1}\,,$ which leads to the following expression for $x_k:$
\begin{equation}
    x_k = (-1)^k\left(\sum_{j=1}^{k-1}(-1)^{j+1}m_j - x_1 \right)\,.
    \label{eq: x_k from x_1}
\end{equation}
Thus, after the measurements, the state of the $n$-qudit system is
\begin{equation*}
    \left| \psi_2 \right>= \frac{1}{\sqrt{p}} \sum_{x_1 \in \mathbb{F}_p} \left| x_1,\,x_1 \right> \bigotimes_{k=2}^{n/2} \left| m_{k-1},\,x_k \right>,
\end{equation*}
where $x_k$ is given by Eq.~\eqref{eq: x_k from x_1} and hides a dependence in $x_1$. It is clear that the application of Pauli corrections of the form $X^{(-1)^k\sum_{j=1}^{k-1} (-1)^j m_j}$ on the unmeasured qudit of each of the measured Bell pairs leaves the system in the state:
\begin{equation*}
    \left| \psi_3 \right> = \frac{1}{\sqrt{p}} \sum_{x_1 \in \mathbb{F}_p} \left| x_1,\,x_1 \right> \bigotimes_{k=2}^{n/2} \left| m_{k-1},\,(-1)^{k+1} x_1 \right>.
\end{equation*}
For $k$ odd, the second qudit of each qudit pair is in the desired state $\left|x_1 \right>$, while for $k$ even, those qudits are in the state $\left| -x_1 \right>$. For qubits, $\left|x_1 \right> = \left|- x_1 \right>,$ but the same is not true for qudits. Therefore, this sign needs to be corrected. The way to do it consists of applying $\textsc{sum}^2$ gates controlled on the second qudit of each odd-numbered qudit pair and targeting the second qudit of the following (even-numbered) pair. Doing so leads to the state:
\begin{equation*}
    \left| \psi_4 \right> = \frac{1}{\sqrt{p}} \sum_{x_1 \in \mathbb{F}_p} \left| x_1,\,x_1 \right> \bigotimes_{k=2}^{n/2} \left| m_{k-1},\, x_1 \right>.
\end{equation*}
Finally, we can reset the measured qudits to $0$ and apply a \textsc{sum} gate controlled on the second qudit of each pair and targeting the first qudit of the ensuing one. This leaves the system in the state:
\begin{equation}
    \left| \mathrm{GHZ}_n \right> = \frac{1}{\sqrt{p}} \sum_{x_1 \in \mathbb{F}_p} \left| x_1 \right>^{\otimes n}\,,
\end{equation}
as promised.

\section{Qudit-exclusive improvements on the practical PBC implementation}\label{app: Improvement of the practical mmts}

Herein, we prove the following Lemma.
\begin{lemma}
    Suppose that the PBC procedure yields a Pauli operator $M=\omega ^{\lambda} X\left( \boldsymbol{x} \right) Z\left( \boldsymbol{z} \right)$ to be measured in quantum hardware. Instead of measuring this operator, we can measure $M^{\prime} = \omega^{k \lambda} X(k\boldsymbol{x}) Z(k\boldsymbol{z})\,,$ with $k\in \mathbb{F}_p \backslash \{0\}$ and $k\boldsymbol{v} = \left( kv_1,\,\ldots,\,kv_t\right) \in \mathbb{F}_p^{t},$ so long as we re-interpret the corresponding measurement outcome, $\sigma^{\prime}$, to get the correct outcome, $\sigma$, of the desired measurement $M$:
    \begin{equation*}
        \sigma \coloneqq \sigma^{\prime} k^{-1} + \boldsymbol{x}\cdot \boldsymbol{z} (k-1)/2\,.
    \end{equation*}
\end{lemma}

\begin{proof}
    The projector associated with the measurement of $M$ with corresponding outcome $\sigma$ is given by:
    \begin{equation*}
    \begin{split}
        \hat{P}_{(M,\sigma)} & = \frac{1}{p} \sum_{j=0}^{p-1} \omega^{-j\sigma} \left[ \omega^{\lambda} X\left( \boldsymbol{x} \right) Z\left( \boldsymbol{z} \right) \right]^j \\
        & = \frac{1}{p} \sum_{j=0}^{p-1} \omega^{-j \left( \sigma - \lambda \right)} \left[X(\boldsymbol{x})Z(\boldsymbol{z}) \right]^j \\
        & = \frac{1}{p} \sum_{j=0}^{p-1} \omega^{-j \left[ \sigma - \lambda - \boldsymbol{x}\cdot \boldsymbol{z} (j-1)/2 \right]} X(j \boldsymbol{x})Z(j \boldsymbol{z})  \,.
    \end{split}
    \end{equation*}

    Because things are computed in modular arithmetic, we can redefine the summation index so that $j \coloneqq j^{\prime}k$. As long as $k\in \mathbb{F}_p\backslash \{0\},$ the range of the sum remains unchanged and we can write:
    \begin{equation*}
    \begin{split}
        \hat{P}_{(M,\sigma)} & = \frac{1}{p} \sum_{j^{\prime}=0}^{p-1} \omega^{-j^{\prime} k \left[ \sigma - \lambda - \boldsymbol{x}\cdot \boldsymbol{z} (j^{\prime} k-1)/2 \right]} X(j^{\prime} k \boldsymbol{x})Z(j^{\prime} k \boldsymbol{z}) \\
        & = \frac{1}{p} \sum_{j^{\prime}=0}^{p-1} \omega^{-j^{\prime} k \left[ \sigma - \boldsymbol{x}\cdot \boldsymbol{z} (k-1)/2 \right]} \left[ \omega^{k \lambda}X(k \boldsymbol{x})Z(k \boldsymbol{z}) \right]^{j^{\prime}} \\
        & = \frac{1}{p} \sum_{j^{\prime}=0}^{p-1} \omega^{-j^{\prime} \sigma^{\prime}} \left[ \omega^{k \lambda} X\left( k\boldsymbol{x} \right) Z\left( k\boldsymbol{z} \right) \right]^{j^{\prime}} \equiv \hat{P}_{(M^{\prime},\sigma^{\prime})}
    \end{split}
    \end{equation*}
    where $\sigma^{\prime} = k \left[ \sigma - \boldsymbol{x}\cdot \boldsymbol{z} (k-1)/2 \right]$ and $M^{\prime} = \omega^{k \lambda} X\left( k\boldsymbol{x} \right) Z\left( k\boldsymbol{z} \right)\,.$ This concludes the proof.
\end{proof}

As explained in Sec.~\ref{subsubsec: Further improvements}, this observation can be used to reduce the number of \textsc{sum} gates that are needed to implement the sequence of Pauli measurements. Note that, once again,  there is an assumption in the proof that we are working with elements of a field, $\mathbb{F}_p$, so that the Lemma holds only for odd-prime numbers.

\section{Generalized stabilizer $\alpha$-Rényi entropies}\label{app: Gen Stab Renyi Entropies}

\subsection{Properties}

In the context of quantum computation with $p$-level systems, the stabilizer $\alpha$-Rényi entropies defined in Ref.~\cite{LeoneOH2022} can be generalized as written in Eq.~\eqref{eq: stabilizer Renyi entropy}. By direct comparison between this definition and the original one, it is rather intuitive that all properties of the original measure still hold. Nevertheless, here we present explicit proofs identical to the ones in the Supplementary Material of Ref.~\cite{LeoneOH2022}. 

\paragraph*{Faithfulness:} $M_{(\alpha,p)} \left( \left| \psi \right> \right) = 0$ iff $\left| \psi \right>$ is a stabilizer state. Otherwise, $M_{(\alpha,p)} \left( \left| \psi \right> \right) > 0\,.$

If $\left| \psi \right>$ is a pure stabilizer state, by definition, there are $p^n$ Pauli operators so that $P \left| \psi \right> = \left| \psi \right>$, for $P\in \mathcal{P}_n$. Therefore, there are $p^n$ Pauli operators $P \in \tilde{\mathcal{P}}_n$ so that $\Xi_P \left( \left| \psi \right>\right) = p^{-n},$ while all other operators in $\tilde{\mathcal{P}}_n$ yield $\Xi_P \left( \left| \psi \right>\right) = 0$ (from the normalization of $\{ \Xi_P \left( \left| \psi \right>\right) \}$).

If $\alpha = 0$, $M_{(0,p)} \left( \left| \psi \right> \right) = \log_p \left( \sum_{P\in \mathcal{S}} 1 \right) - n = n-n = 0,$ where $\mathcal{S}$ is the set of $p^n$ operators with non-zero contribution.

For any other $\alpha \neq 1$, we note that $M_{(\alpha,p)} \left( \left| \psi \right> \right) = \frac{1}{1-\alpha} \log_p \sum_{P\in \mathcal{S}} \Xi_P^{\alpha} \left( \left| \psi \right>\right) - n = \frac{1}{1-\alpha} \log_p p^{(1-\alpha)n} -n = 0.$

The result for $\alpha =1$ follows from continuity.

Now, we need to prove the converse direction. Thus, suppose that $M_{(\alpha,p)} \left( \left| \psi \right> \right) = 0\,.$ This implies that $\sum_{P\in \tilde{\mathcal{P}}_n} \left| \left< \psi \left| P \right| \psi \right> \right|^{2\alpha} = p^n\,.$
At the same time, the fact that $\left| \psi \right>$ is a pure state implies that $\sum_{P\in \tilde{\mathcal{P}}_n} \left| \left< \psi \left| P \right| \psi \right> \right|^{2} = p^n\,.$

Take $\alpha = 0.$ This turns the first condition into $\sum_{P\in \tilde{\mathcal{P}}_n} 1 = p^n\,,$ which basically tells us that there are exactly $p^n$ Pauli operators which yield a non-zero expectation value. The purity condition further informs us that these contributing elements necessarily have $\left| \left< \psi \left| P \right| \psi \right> \right|^{2} = 1\,.$ Therefore, $\left| \psi \right>$ is necessarily a stabilizer state. 

For any other $\alpha \neq 1,$ the compatibility between the two conditions leads to $\sum_{P\in \tilde{\mathcal{P}}_n} a_{P}^{\alpha} = \sum_{P\in \tilde{\mathcal{P}}_n} a_{P},$ where $a_P = \left| \left< \psi \left| P \right| \psi \right> \right|^{2} \in \mathbb{R}_0^+\,.$ Necessarily, this implies that either $a_P = 0$ or $a_P = 1\,.$ For the two conditions to hold, there are $p^n$ Pauli operators such that $a_P =1$ which, just as for $\alpha = 0$, means that $\left| \psi \right>$ is a stabilizer state.

Once again, the result for $\alpha = 1$ follows by continuity.

To prove that $M_{(\alpha, p) } \left( \left| \psi \right> \right) > 0$ for any non-stabilizer state, we note that the hierarchy of $\alpha$-Rényi entropies is such that $M_{(\alpha, p)}\left( \left| \psi \right> \right) \geq M_{(\alpha + a, p)}\left( \left| \psi \right> \right),$ $\forall a > 0\,.$ Further, we note that in the limit where $\alpha \rightarrow \infty$ we have: $\lim_{\alpha \rightarrow \infty} M_{(\alpha, p)}\left( \left| \psi \right> \right) = 0.$ This means that for any finite $\alpha$ we have $M_{(\alpha, p)} \left( \left| \psi \right> \right) \geq 0.$ Since the equality holds only for stabilizer states, as proved above, necessarily the $\alpha$-Rényi entropies are positive for all other states. 

\paragraph*{Stability under free operations:} Let $C$ be a free operation; then, $M_{(\alpha,p)} \left( C \left| \psi \right> \right) = M_{(\alpha,p)} \left( \left| \psi \right> \right).$

In this context, $C$ being a free operation means that it belongs to the Clifford group on $n$ qudits. Any Pauli operator $P \in \tilde{\mathcal{P}}_n$ is mapped under conjugation by $C$ into another Pauli operator: $C^{\dagger} P C = \omega^{\varphi} Q,$ with $Q \in \tilde{\mathcal{P}}_n$ and $\varphi\in \mathbb{F}_p$. With this observation, the proof of this property is trivial:
\begin{equation*}
\begin{split}
    \Xi_P \left( C \left| \psi \right> \right) & = \frac{1}{p^n} \left| \left< \psi \left| C^{\dagger} P C \right| \psi \right> \right|^{2} \\
    & = \frac{1}{p^n} \left| \left< \psi \left| Q \right| \psi \right> \right|^{2} = \Xi_Q \left( \left| \psi \right> \right),
\end{split}
\end{equation*}
which leads to
\begin{equation*}
\begin{split}
    M_{(\alpha, p)} \left( C \left| \psi \right> \right) & = \frac{1}{1-\alpha} \log_p \left( \sum_{P\in \tilde{\mathcal{P}}_n} \Xi_{P}^{\alpha} \left( C \left| \psi \right> \right) \right) - n \\
    & = \frac{1}{1-\alpha} \log_p \left( \sum_{Q\in \tilde{\mathcal{P}}_n}  \Xi_Q^{\alpha} \left( \left| \psi \right> \right) \right) - n \\
    & = M_{(\alpha, p)} \left( \left| \psi \right> \right)\,.
\end{split}
\end{equation*}

\paragraph*{Additivity:} $M_{(\alpha,p)} \left( \left| \psi_1 \right> \otimes \left| \psi_2 \right> \right) = M_{(\alpha,p)} \left( \left| \psi_1 \right> \right) + M_{(\alpha,p)} \left( \left| \psi_2 \right> \right).$

The proof of this final property is rather straightforward. Suppose that $\left| \psi_1 \right>$ (resp. $\left| \psi_2 \right>$) is a quantum state of $n_1$ (resp. $n_2$) qudits, so that $\left| \Psi \right> = \left| \psi_1 \right> \otimes \left| \psi_2 \right>$ is a quantum state of $n=n_1 + n_2$ qudits. Then, we can write:
\begin{equation*}
    M_{(\alpha,p)} \left( \left| \Psi \right> \right) = \frac{1}{1-\alpha} \log_p \left( \sum_{P\in \tilde{\mathcal{P}}_n} \Xi_{P}^{\alpha} \left( \left|  \Psi \right> \right) \right) - n\,.
\end{equation*}
Any Pauli operator $P\in \tilde{\mathcal{P}}_n$ can be written as $P=P_1 \otimes P_2$, with $P_1\in \tilde{\mathcal{P}}_{n_1}$ and $P_2\in \tilde{\mathcal{P}}_{n_2}\,.$ Thus, we note that $\Xi_{P} \left( \left| \Psi \right> \right) = \Xi_{P_1} \left( \left| \psi_1 \right> \right) \Xi_{P_2} \left( \left| \psi_2 \right> \right)\,,$ which allows us to re-write the equation above as
\begin{widetext}
\begin{equation*}
\begin{split}
    M_{(\alpha,p)} \left( \left|  \Psi \right> \right) & = \frac{1}{1-\alpha} \log_p \left( \sum_{P_1\in \tilde{\mathcal{P}}_{n_1}} \Xi_{P_1}^{\alpha} \left( \left| \psi_1 \right> \right)
    \sum_{P_2\in \tilde{\mathcal{P}}_{n_2}} \Xi_{P_2}^{\alpha} \left( \left| \psi_2 \right> \right)\right) - (n_1 + n_2) \\
    & = \frac{1}{1-\alpha} \log_p \left( \sum_{P_1\in \tilde{\mathcal{P}}_{n_1}} \Xi_{P_1}^{\alpha} \left( \left| \psi_1 \right> \right) \right) - n_1 + \frac{1}{1-\alpha} \log_p \left( \sum_{P_2\in \tilde{\mathcal{P}}_{n_2}} \Xi_{P_2}^{\alpha} \left( \left| \psi_2 \right> \right)\right) - n_2 \\
    & = M_{(\alpha,p)} \left( \left| \psi_1 \right> \right) + M_{(\alpha,p)} \left( \left| \psi_2 \right> \right)\,.
\end{split}
\end{equation*}
\end{widetext}
This property makes it easy to compute the stabilizer $\alpha$-Rényi entropies of $k$ copies of non-stabilizer states, as one notes that $M_{(\alpha,p)} \left( \left| \psi_1 \right>^{\otimes k} \right) = k M_{(\alpha,p)} \left( \left| \psi_1 \right> \right).$ This will be useful in the ensuing calculations allowing us to set a lower bound on the RoM of qutrit and ququint magic states.

\subsection{Hybrid computation lower bounds}

Here we detail the calculation of the lower bounds on the cost of hybrid computation with qutrits and ququints presented in Sec.~\ref{subsec: Hybrid computation}. From Eq.~\eqref{eq: Hybrid computation cost lower bound} together with the additivity of $\alpha$-Rényi entropies we note that:
\begin{equation}
    \mathcal{R}^{2} \left( \left| T_v \right>^{\otimes k} \right) \geq p^{M_{\left(\nicefrac{1}{2},p\right)}\left( \left| T_v \right>^{\otimes k} \right)} = p^{k M_{\left(\nicefrac{1}{2},p\right)}\left( \left| T_v \right> \right)}\,.
\end{equation}

The $\nicefrac{1}{2}$-Rényi entropy of a single-qudit magic state $\left| T_v \right>$ can be written as:
\begin{equation}
    M_{\left(\nicefrac{1}{2},p\right)}\left( \left| T_v \right> \right) = 2 \log_p \left( \frac{1}{p} \sum_{P\in \tilde{\mathcal{P}}_1} \left| \left< T_v \left| P \right| T_v \right> \right| \right)\,.
    \label{eq: 1/2 Renyi entropy}
\end{equation}

As stated in Sec.~\ref{subsec: Universality via MSI}, we consider the single-qutrit magic state $\left| T_v \right> = \left( \left| 0 \right> + e^{2\pi i /9} \left| 1 \right> + e^{-2\pi i / 9} \left| 2 \right>  \right)/\sqrt{3}\,.$ The sum in Eq.~\eqref{eq: 1/2 Renyi entropy} runs over the nine single-qutrit Pauli operators with phase equal to $1$ and yields the result $\sum_{P\in \tilde{\mathcal{P}}_1} \left| \left< T_v \left| P \right| T_v \right> \right| = (1 + 2\sqrt{3})\,.$ This immediately leads to a lower bound for the cost of hybrid computation with qutrits given by:
\begin{equation}
    \mathcal{R}^{2} \left( \left| T_v \right>^{\otimes k} \right) \geq 3^{k2\log_3 \left( \frac{1+2\sqrt{3}}{3} \right) } = \left( \frac{1+2\sqrt{3}}{3} \right)^{2k}\,,
\end{equation}
which corresponds to the lower bound given in the main text: $\Omega \left( 3^{0.7236 k} \epsilon^{-2} \right)$.

The single-ququint magic state is given by $\left| T_v \right> = \frac{1}{\sqrt{5}} \left( \left| 0 \right> +  e^{-4\pi i /5}  \left| 1 \right> + e^{-2\pi i / 5} \left| 2 \right> + e^{4\pi i /5} \left| 3 \right> + e^{2\pi i / 5} \left| 4 \right> \right)$.
In this case, the sum in Eq.~\eqref{eq: 1/2 Renyi entropy} involves 25 Pauli operators and has a value of $\left( 1+4\sqrt{5} \right).$ This immediately leads to the lower bound for the cost of hybrid computation with ququints given by:
\begin{equation}
    \mathcal{R}^{2} \left( \left| T_v \right>^{\otimes k} \right) \geq 5^{k2\log_5 \left( \frac{1+4\sqrt{5}}{5} \right) } = \left( \frac{1+4\sqrt{5}}{5} \right)^{2k}\,,
\end{equation}
which can be presented as $\Omega \left( 5^{0.8544 k} \epsilon^{-2} \right).$

%\bibliography{bibliography.bib}% Produces the bibliography via BibTeX.

\providecommand{\noopsort}[1]{}\providecommand{\singleletter}[1]{#1}%

\end{document}